\RequirePackage{fix-cm}
\documentclass[smallextended,nospthms]{svjour3}  
\smartqed  
\usepackage{graphicx}
\usepackage{url}
\usepackage{amsmath,amssymb,amsthm,amsfonts}
\newtheorem{theo}{Theorem}

\def\Rset{\mathbb{R}}

\newcommand{\comment}[1]{}
\usepackage{textcomp}
\usepackage{natbib}
\usepackage{color}

\begin{document}
\title{Behavioral modulation of the coexistence between \textit{Apis melifera} 
and \textit{Varroa  destructor}: A defense against colony colapse disorder?}

\titlerunning{Behavioral Defense against Colony Colapse Disorder}        

\author{Joyce de Figueir\'o Santos \and
        Fl\'avio Code\c{c}o Coelho \and  
        Pierre Alexandre Bliman
}


\institute{Joyce de Figueir\'o Santos and Fl\'avio Code\c{c}o Coelho and        
Pierre Alexandre Bliman \at
              Applied Mathematics School -- Getulio Vargas Foundation, Brazil \\
              Tel.: +55-21-37995551\\
              Fax: +55-21-37995917\\
              \email{fccoelho@fgv.br}           
            \and
            Pierre Alexandre Bliman, also at \at
               Inria, France
}

\date{Received: date / Accepted: date}

\maketitle

\begin{abstract}
 
Colony Collapse Disorder has become a global problem for beekeepers and for 
the crops which depend on bee polination. Multiple factors are known to 
increase the risk of colony colapse, and the ectoparasitic mite \textit{Varroa 
destructor} that parasitizes honey bees is among the main threats to colony 
health.
Although this mite is unlikely to, by itself, cause the collapse of hives, it 
plays an important role as it is a vector for many viral diseases. Such diseases 
are among the likely causes for Colony Collapse Disorder.

The effects of \textit{\textit{V. destructor}} infestation are disparate in 
different 
parts of the world. 
Greater morbidity - in the form of colony losses -  has been reported in 
colonies of European honey bees (EHB) in Europe, Asia and North America. 
However, this mite has been present in Brasil for many years and yet there are 
no reports of Africanized honey bee (AHB) colonies losses.

Studies carried out in Mexico showed that some resistance behaviors to the mite 
- especially grooming 
and hygienic behavior - appear to be different in each subspecies. Could those 
difference in behaviors 
explain why the AHB are less susceptible to Colony Collapse Disorder?   

In order to answer this question, we propose a mathematical model of the 
coexistence dynamics of these two species, the bee and the mite, 
to analyze the role of resistance behaviors in the overall health of the 
colony, 
and, as a consequence, its ability to face epidemiological challenges.

\end{abstract}

\section{Introduction}

Since 2007 American beekeepers reported heavier and widespread losses of bee 
colonies. And this goes beyond American borders --- many Europeans beekeepers 
complain of the same problem. This mysterious phenomenon was called "Colony 
Collapse Disorder" (CCD)  --- the official description of a syndrome in which 
many bee colonies died in the winter and spring of 2006/2007. Diseases and 
parasites, in-hive chemicals, agricultural insecticides, genetically modified 
crops, changed cultural practices and cool brood are pointed as some of the 
possible 
causes for CCD \citep{oldroyd_whats_2007}.

The ectoparasitic mite \textit{Varroa destructor} that parasitize honey bees has 
become 
a global problem and is considered as one of the important burdens on bee 
colonies and a cause for CCD.
The Varroa mite is suspected of having caused the collapse of millions of 
Apis mellifera honey bee colonies 
worldwide. However, the effects caused by \textit{V. destructor} infestation 
vary in 
different parts  of the world. More intense losses have been reported in 
European honey bee colonies (EHB) of Europe, Asia and North America 
\citep{calderon_reproductive_2010}.

The life cycle of \textit{V. destructor} is tightly linked with the bee's. 
Immature 
mites develop together with immature bees, parasitizing them from an early 
stage. The mite's egg-laying behavior is coupled with the bee's and thus 
depends on its reproductive cycle. Since worker brood rearing and thus 
Varroa reproduction occurs all year round in tropical 
climates, it could be expected that the impact of the parasite would be even 
worse in tropical regions. But \textit{\textit{Varroa destructor}} has been 
present in 
Brazil for more than 30 years and yet no 
collapses due to this mite, have been recorded \citep{carneiro_changes_2007}. It 
is worth noting that the dominant variety of bees in Brazil is the Africanized 
honey bee (AHB) which since its introduction in 1956, has spread to the entire 
country\citep{pinto_ectoparasite_2012}.

African bees and their hybrids are 
more resistant to the mite \textit{V.destructor} than European bee 
subspecies \citep{medina_comparative_1999, pinto_ectoparasite_2012}. A review by 
\cite{arechavaleta-velasco_relative_2001} 
in Mexico showed that EHB was 
twice as attractive to V.destructor than AHB. The removal of naturally 
infested brood, which is termed hygienic behavior, was reported as four times 
higher in AHB than in EHB, and AHB workers were more efficient in grooming mites 
from their bodies. 

These  behaviors are important factors  in keeping the mites infestation low in 
the honey bee colonies.

\subsection{Resistance behaviors of the bee against the parasite}

Two main resistance behaviors, namely grooming and hygienic 
behavior\citep{spivak_honey_1996}, are 
mechanisms employed by the honey bees to control parasitism in the hive.

The grooming behavior is when a worker bee is able to groom herself with her 
legs and mandibles to remove the mite and then injure or kill it.
\citep{vandame_levels_2000}.

Hygienic behavior is a mechanism through which worker bee broods are uncapped 
leading to the death of the pupae. This behavior is believed to 
confer resistance to Varroa infestation since worker bees are more likely to 
uncap an infested brood, than an uninfested one. It has been demonstrated that 
the smell of the mite by itself is capable of activating this behavior. 
\citep{correa-marques_uncapping_1998}.

The hygienic behavior serves to combat other illnesses or parasites to which 
the brood is susceptible. It is also not a completely accurate mechanism. 
Correa-Marques and De Jong (\citeyear{correa-marques_uncapping_1998}), report 
that the majority (53\%) of the uncapped cells display apparently no signs of 
parasitism 
or other abnormality which would justify the killing of the brood. Thus, in our 
model we define two parameters for the hygienic behavior: $H_g$ , for the 
generic hygienic behavior, which may kill uninfested pupae, and $h$ for the 
sucess rate in uncapping infested brood cells.

Africanized honey bees have been shown to be more competent in hygienic 
behavior than European honey bees. \citet{vandame_levels_2000} found in Mexico 
that the EHB are able to remove just $8\%$ of infested brood while AHB removed 
up 32.5$\%$.

The main goal of this paper is to propose a model capable of describing 
the dynamics of infestation by V. destructor in bee colonies taking into 
consideration bee's resistance mechanisms to mite infestation --- 
 grooming and hygienic behavior. In addition, through simulations, we show 
how the resistance behaviors contribute to the reduction infestation levels and 
may even lead to the complete elimination of the parasite from the colony.
    
\section{Mathematical model}

 \begin{figure}[!ht]
  \begin{center}
  \includegraphics[scale=0.4]{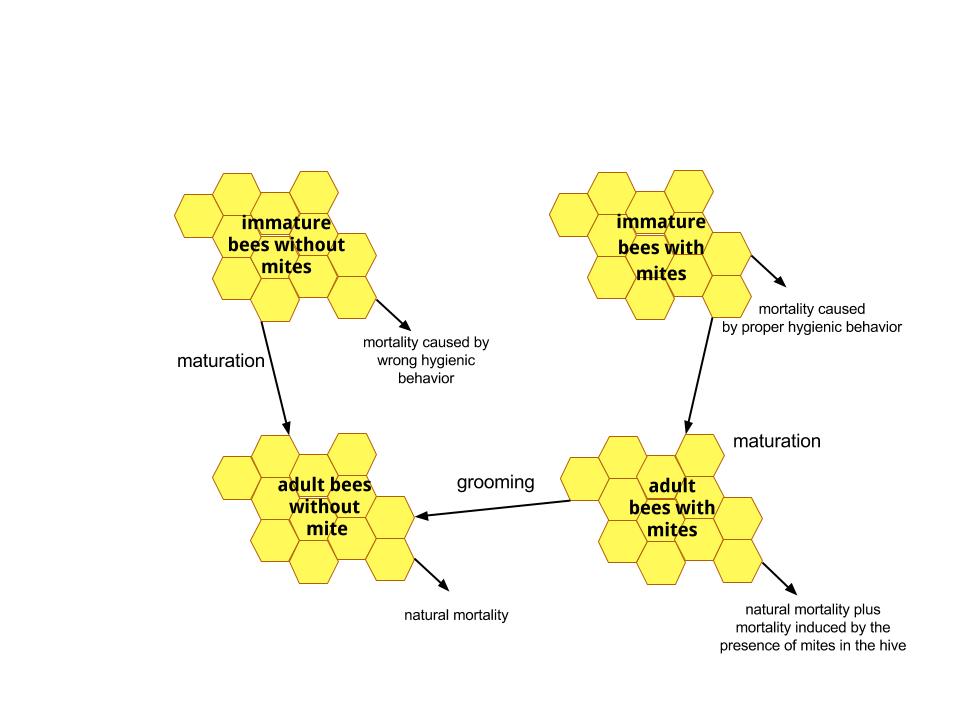} 
  \caption{Diagram to describes the dynamics of the model.}
  \label{fig_diagrama}
  \end{center}  
   
  \end{figure}

Previous work by \cite{ratti_mathematical_2012} models the population dynamics 
of bee and mites together with the acute bee paralysis virus.  
Here we focus solely on the host-parasite interactions trying to 
understand the resilience of colonies in Brazil and the role of the more 
efficient resistance behaviors displayed by AHB to explain the lower 
infestation rates and incidence of collapses in their colonies.

\cite{vandame_parasitism_2002} discusses the cost-benefit of resistance 
mechanism of bee against mite. The \textit{grooming} behavior 
performed by adult bees, includes detecting and eliminating mites from their 
own body (auto-grooming) or from the body of another bee (allo-grooming). The 
hygienic behavior occurs when adult bees detect 
the presence of the mite offspring still in the cells and in order to prevent 
the mites 
from spreading in the colony, the worker bees kill the infested brood.
Their study compared the results for two subspecies of bees - Africanized and 
European - to examine whether these two mechanisms could explain the observed 
low compatibility between Africanized bees and the mite \textit{Varroa 
destructor}, in Mexico.
The results showed that \textit{grooming} and hygienic behavior appears most 
intense in Africanized bees than in Europeans bees.

The model proposed is shown in the diagram of figure \ref{fig_diagrama}, and 
detailed in the system of differential equations below:

\begin{align}
\label{eq:modelo}
\dot{I} &= \pi \frac{A}{A+A_i} - \delta I - H I \nonumber \\
\dot{A} &= \delta I + gA_i - \mu A \nonumber \\
\dot{I_i} &= \pi \frac{A_i}{A+A_i} - \delta I_i - H_i I_i \nonumber \\
\dot{A_i} &= \delta I_i - gA_i - (\mu + \gamma) A_i
\end{align}

In the proposed model, {\bf{$I$}}, {\bf{$I_{i}$}}, {\bf{$A$}} and 
{\bf{$A_{i}$}} 
 represent the non-infested immature 
bees, infested immature bees, non-infested adult worker bees 
and infested adult worker bees, respectively. 

Daily birth rate for bees is denoted by $\pi$, 
$\delta$ is the maturation rate, i.e., the inverse of number of days 
an immature bee requires to turn in adult, this rate is the same for  
both infested and  non-infested immature bees. $\mu$ is the mortality rate for 
adult bees,$\gamma$ is the mortality rate induced by the presence of mites in 
the colony bees.
The parameters $H_i$, $H$ e $g$ are the rate of removal of infested pupae via 
hygienic behavior, the general 
hygienic rate (affecting uninfested pupae)  and grooming rate, respectively.

 \begin{table}[h]

\caption{Parameters of the model.}
\centering

\label{tab_parametros}

\begin{small} 
\begin{tabular}{c|c c c c}

\hline
\textbf{Parameters} & \textbf{Meaning} & \textbf{Value} & \textbf{Unit} &  
\textbf{Reference} \\
\hline
$\pi$ & Bee daily birth rate & 2500 & $bees \times day^{-1}$ & 
\citealp{Embrapa} 
 \\
$\delta$ & Maturation rate & $0.05$ & $day^{-1}$ & \citealp{Embrapa} \\
$H$ & Generic hygienic behavior & - & $day^{-1}$ & - \\
$H_i$ & Hygienic behavior towards infested brood   & - & $day^{-1}$ & - \\
$g$ & \textit{Grooming} & - & $day^{-1}$ & - \\
$\mu$ & Mortality rate & $0.04$ & $day^{-1}$ & \citep{khoury_quantitative_2011} 
\\
$\gamma$ & Mite induced mortality & $10^{-7}$ & $day^{-1}$ & 
\citep{ratti_mathematical_2012}  \\

\hline
\end{tabular}
\end{small}
\end{table}

\subsection*{Choosing parameters}

Some of the parameters associated with the bees life 
cycle, used for the simulations,  
can be found in the literature, as shown in table \ref{tab_parametros}. For the 
resistance behavior parameters, $g$, $H$ and $H_i$, very little information is 
available. 
Therefore we decided to study the variation of these parameters within ranges 
which allowed for the system to switch between a mite-free equilibrium to one 
of coexistence.
These ranges also reflected observations described in the 
literature \citep{mondragon_multifactorial_2005, 
vandame_parasitism_2002, arechavaleta-velasco_relative_2001}.

\begin{table}[h]

\caption{Varying the parameters}
\centering

\label{tab_parametros_var}

\begin{small} 
\begin{tabular}{c|c c}
\hline
\textbf{Parameter} & \textbf{Maximum value} & \textbf{Minimum value} \\
\hline
$g$ & $0.01$ & $0.1$  \\
$H_i$ & $0.08$ & $0.4$ \\
$H$ & $0.04$ & $0.2$ \\
\hline
 
\end{tabular}
\end{small}
\end{table}

 The three unknown parameters representing resistance behaviors $g$, $H_i$, 
$H$ -- grooming,  proper hygienic behavior and wrong hygienic behavior --   
where studied with respect to the existence of a coexistence equilibrium.

\section{Results}

In order to understand the dynamics of the proposed model of mite infestation 
of bee colonies, we proceed to analyze it. 

\subsection{Basic reproduction number of the infested bees }

 An effective way to look at boundary beyond which coexistence of mites and 
bees is possible, is to look at the ${\cal R}_0$ of infestation. 
For our model, the basic reproduction number, or ${\cal R}_0$ of infested bees, 
can 
be thought of as the number of new infestations that one infested bee when 
introduced into the colony generates on average over the course of its 
infestation period or while it is not groomed, in an otherwise uninfested 
population.

 \begin{figure}
  \begin{center}
   \includegraphics[scale=0.5]{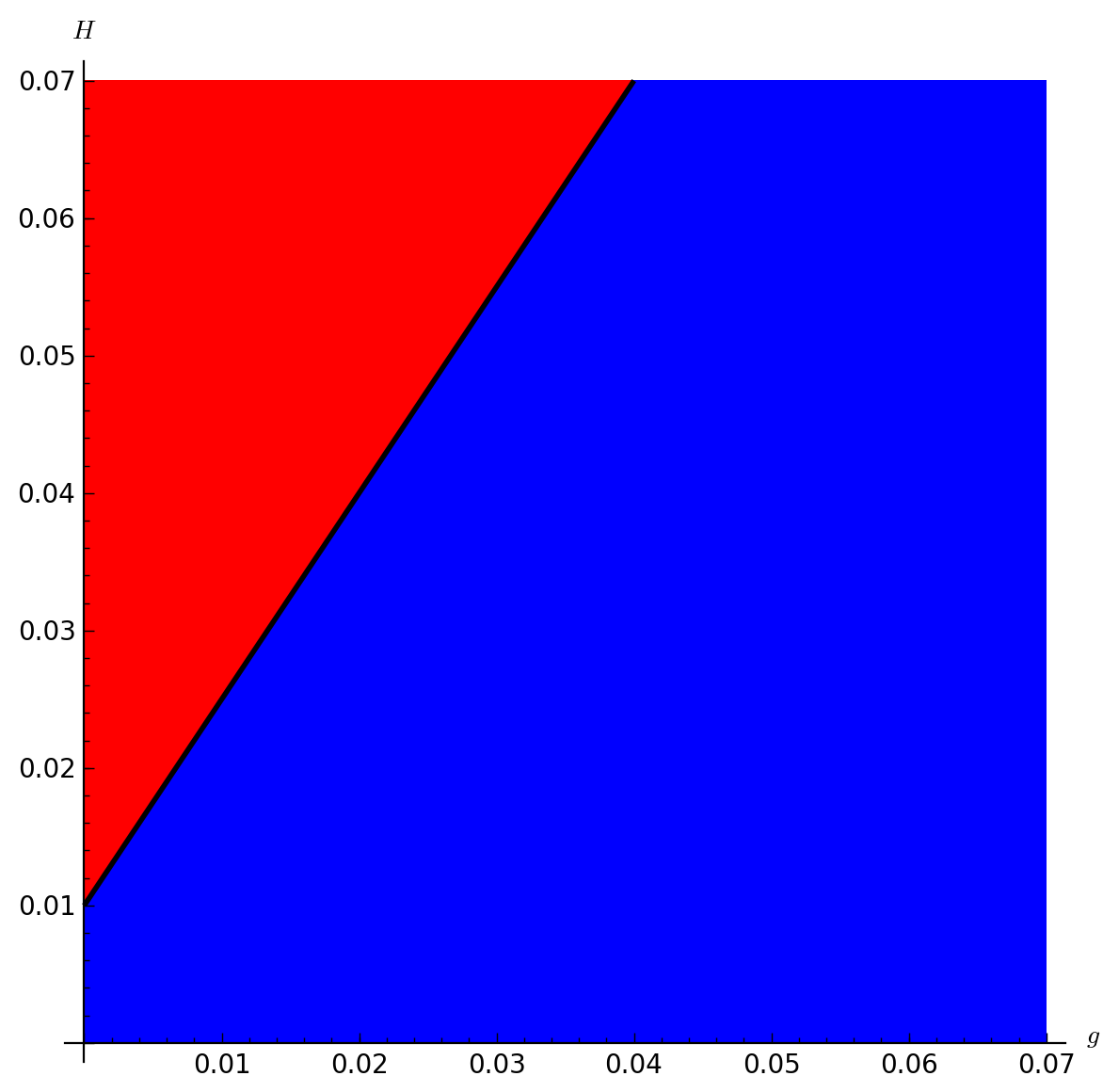} 
  \caption{plot of values of ${\cal R}_0$ for a range of values of $g$ and $H$. 
$H_i=0.01$ and remaining parameters set as described in table 
\ref{tab_parametros}. The region in red corresponds to ${\cal R}_0>1$, the 
black line 
to ${\cal R}_0=0$ and the blue region otherwise.}
  \label{fig_gxhi}
  \end{center}  
   
  \end{figure}

\begin{figure}
  \begin{center}
   \includegraphics[scale=0.5]{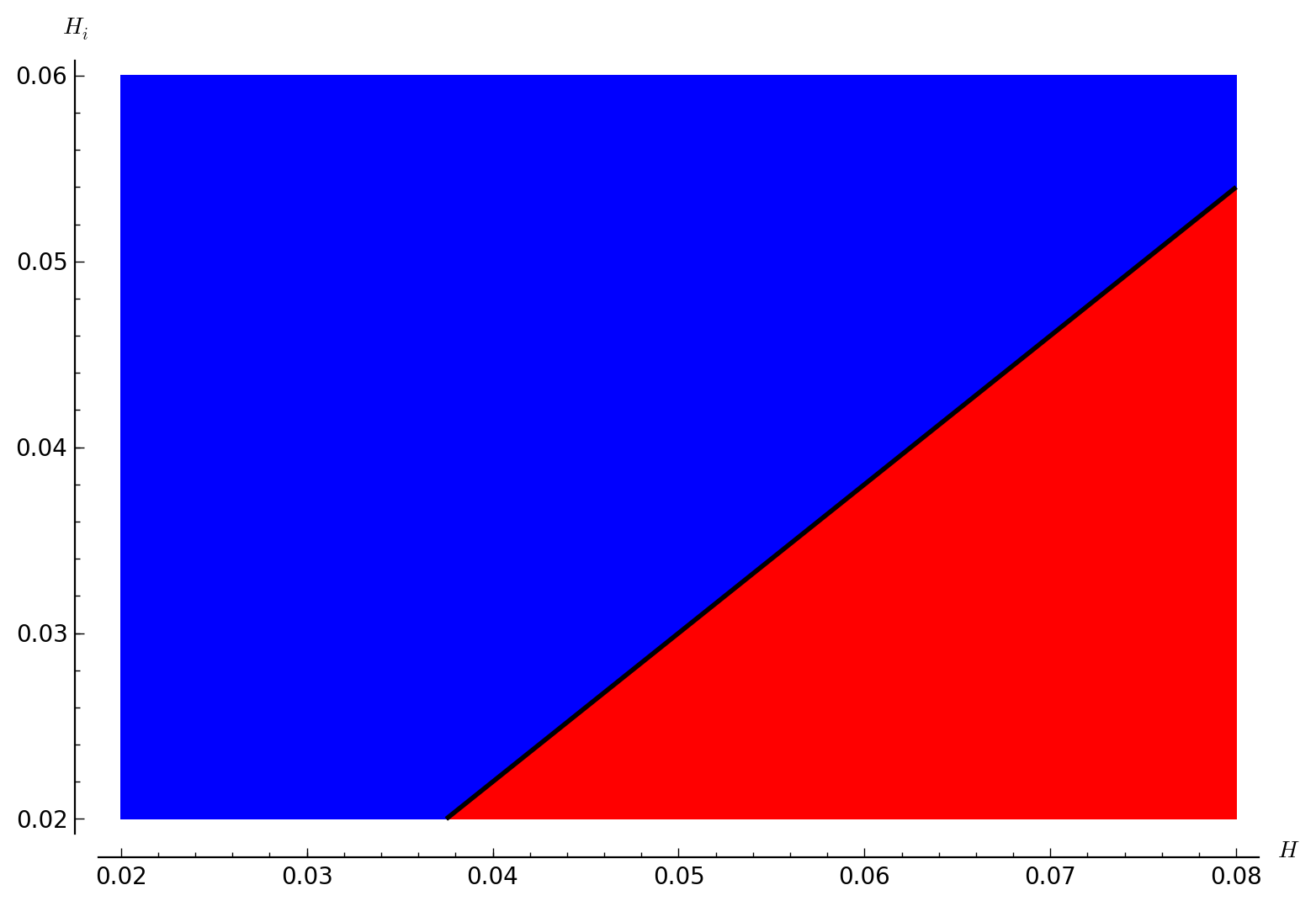} 
  \caption{Values of ${\cal R}_0$ for various combinations of $H_i$ and $H$. 
$g=0.01$ and other parameters as given in table \ref{tab_parametros}.The 
region 
in red corresponds to ${\cal R}_0>1$, the black line 
to ${\cal R}_0=0$ and the blue region otherwise. This figure illustrates one of 
the 
conditions for coexistence(given other parameters values fixed as in table 
\ref{tab_parametros}) that $H$ must be larger than $H_i$. }
  \label{fig_hxhi}
  \end{center}  
   
  \end{figure}

\paragraph{Deriving ${\cal R}_{0}$ using the next generation method: } 

To calculate the basic reproduction number of infested bees, we will use the 
next-generation matrix \citep{van_den_driessche_reproduction_2002},
where the whole population is divided into $n$ compartments in which there are 
$m < n$ infested compartments.

In this method, ${\cal R}_{0}$ is defined as the spectral radius, or the 
largest 
eigenvalue, of the next generation matrix.

Let $x_{i}$, $i=1,2...,m$ be the number or proportion of individuals in the
$ith$ compartment. Then
\begin{center}
$\frac{dx_{i}}{dt} = \mathcal{F}_{i}(x) - \mathcal{V}_{i}(x)$
\end{center}
where $\mathcal{F}_{i}(x)$ is the rate of appearance of new infections in 
compartment
$i$ and $\mathcal V_{i}(x) = \mathcal V_{i}^{-}(x) -\mathcal V_{i}^{+}(x)$. 
Where $\mathcal V_{i}^{-}$ is 
the rate of transfer of individuals out of the
$ith$ compartment, and $\mathcal V_{i}^{+}$ represents the rate of transfer of 
individuals into compartment $i$ by all other means.

The next generation matrix is then defined by $FV^{-1}$, where $F$ and $V$
can be formed by the partial derivatives of $\mathcal F{i}$ and $\mathcal 
V_{i}$.

\begin{center}
$F = [\frac{\partial\mathcal F_{i}(x_{0})}{\partial x_{j}}]$ and $V = 
[\frac{\partial\mathcal V_{i}(x_{0})}{\partial x_{j}}]$
\end{center}

where $x_{0}$ is the disease free equilibrium.

In our model, $m=2$ and the infested compartments are:

\begin{align}
\label{eq:r0eq}
\frac{dI_{i}}{dt} &= \pi \frac{Ai}{A+Ai} - \delta I_{i} - HI_{i} \nonumber \\
\frac{dAi}{dt} &= \delta I_{i} - gAi - (\mu + \gamma) Ai
\end{align}

Now we write the matrices F and V , substituting the mite-free equilibrium 
values, $A^*=\frac{\delta \pi}{\mu(\delta+H)}$ and $A_i^*=0$. 

\[
F = \left[\begin{array}{cc}
	      0 & \frac{\mu(\delta+H}{\delta}) \\
	      0 & 0 \\
	  \end{array}
          \right] \]\label{matrixF}

\[ V = \left[\begin{array}{cc}
\delta + H_i & 0 \\
-\delta & g + \gamma + \mu \\
\end{array}
          \right] \]\label{matrixV}

Let the next-generation matrix $G$ be the matrix product $FV^{-1}$. Then 

\[G = \left[\begin{array}{cc}
\frac{\mu(\delta +H)}{(\delta + H_i)(g + \gamma + \mu)} & \frac{\mu(\delta 
+H)}{\delta(g + \gamma + \mu)} \\
0  & 0 \\
\end{array}
          \right] \] \label{matrixG}

Now we can find the basic reproduction number, ${\cal R}_{0}$, which is the 
largest eigenvalue of the matrix $G$.

\begin{align}
 \label{r0}
 {\cal R}_0 = \frac{\mu(\delta +H)}{(\delta + H_i)(g + \gamma + \mu)} 
\end{align}

 \begin{figure}[h]
  \begin{center}
  \includegraphics[scale=0.7]{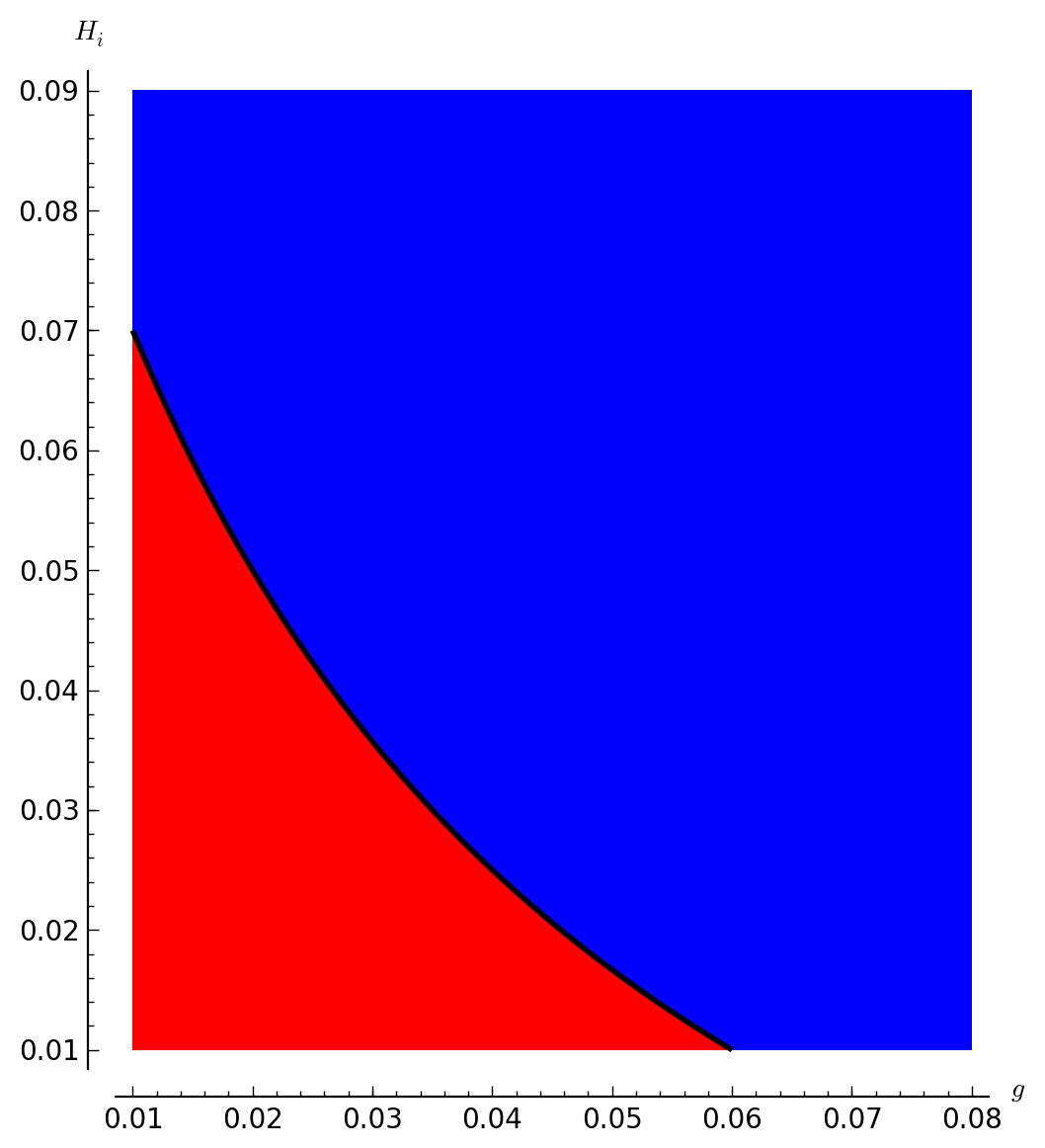} 
  \caption{Implicit plot for ${\cal R}_0$ letting $g$ and $H_i$ vary. 
  Using the values for parameters $\pi$, $\delta$, $\mu$ and $\gamma$
  from table \ref{tab_parametros} 
  The red region represent ${\cal R}_0>1$ which means that for these 
combination of $g$ 
and $H_i$ the mite will stay in the colony. 
  On the other hand, the blue region represents ${\cal R}_0<1$ which means that 
for 
these these combination of $g$ and $H_i$
  the mites will be eliminated.}
  \label{fig_R0}
  \end{center}  
   
  \end{figure} 
  
Figures \ref{fig_gxhi}, \ref{fig_hxhi} and \ref{fig_R0} show the boundary 
between mite-free (blue region, ${\cal R}_0<1$) and coexistence equilibria (red 
region, ${\cal R}_0>1$). 

\subsection{Well-Posed and Boundedness}

For sake of simplicity, we denote
\begin{equation}
\alpha \doteq \delta +H,\qquad \alpha_i \doteq \delta +H_i,\qquad
\mu_i \doteq \mu+\gamma
\end{equation}
in such a way that the system (\ref{eq:modelo}) rewrites
\begin{subequations}
\label{eq0}
\begin{gather}
\label{eq0a}
\dot I = \pi \frac{A}{A+A_i} - \alpha I\\
\label{eq0b}
\dot A = \delta I -\mu A + gA_i \\
\label{eq0c}
\dot I_i = \pi \frac{A_i}{A+A_i} - \alpha_iI_i\\
\label{eq0d}
\dot A_i = \delta I_i -(\mu_i+g)A_i
\end{gather}
\end{subequations}
We assume that all the coefficients presented in table \ref{tab_parametros} are 
all positive, that is:
\begin{equation}
\label{eq13}
\pi, \delta, \mu >0,\qquad \alpha, \alpha_i > \delta, \qquad \mu_i > \mu\ .
\end{equation}
The previous system of equations is written
\begin{equation}
\label{eq1}
\dot X = f(X),\qquad X=(I, A, I_i, A_i)
\end{equation}
The right-hand side of \eqref{eq1} is not properly defined in the points where 
$A+A_i=0$.
However, the following result demonstrates that this has no consequence on the 
solutions, as the latter stays away from this part of the subspace.
For subsequent use, we denote $\cal D$ the subset of those elements $X=(I, A, 
I_i, A_i)\in\Rset_+^4$ such that $A+A_i\neq 0$.

\begin{theo}[Well-posedness and boundedness]
\label{le1}
If $X_0\in\cal D$, then there exists a unique solution of \eqref{eq1} defined 
on 
$[0,+\infty)$ such that $X(0)=X_0$.
Moreover, for any $t>0$, $X(t)\in\cal D$, and
\begin{subequations}
\label{eq3}
\begin{gather}
\label{eq3a}
\frac{\pi}{\alpha_{\max}}
\leq \liminf_{t\to +\infty} (I(t)+I_i(t))
\leq \limsup_{t\to +\infty} (I(t)+I_i(t))
\leq \frac{\pi}{\alpha_{\min}}\\
\label{eq3b}
\frac{\delta\pi}{\mu_i\alpha_{\max}}
\leq \liminf_{t\to +\infty} (A(t)+A_i(t))
\leq \limsup_{t\to +\infty} (A(t)+A_i(t))
\leq \frac{\delta\pi}{\mu\alpha_{\min}}
\end{gather}
\end{subequations}
where by definition $\alpha_{\min} \doteq \min\{\alpha;\alpha_i\}$, 
$\alpha_{\max} \doteq \max\{\alpha;\alpha_i\}$.
Also,
\begin{equation}
\label{eq03}
\frac{1}{(\alpha-\alpha_{\min})\mu+\alpha g}
\frac{\pi g\mu\alpha_{\min}}{\mu_i\alpha_{\max}} \leq \liminf_{t\to +\infty} 
I(t) ,\qquad
\frac{1}{(\alpha-\alpha_{\min})\mu+\alpha g}
\frac{\delta\pi g\alpha}{\mu_i\alpha_{\max}} \leq \liminf_{t\to +\infty} A(t)
\end{equation}
and
\begin{equation}
\label{eq03a}
(I_i(0), A_i(0))\neq (0,0) \Rightarrow \forall t\geq 0,\ I_i(t)>0,\ A_i(t) >0
\end{equation}
\end{theo}

Define ${\cal D}'$ as the largest set included in $\cal D$ and fulfilling the 
inequalities of Theorem \ref{le1}, that is:
\begin{multline}
\label{eq19}
{\cal D}' \doteq \left\{
(I, A, I_i, A_i)\in\Rset_+^4\ :\ \frac{\pi 
g\mu\alpha_{\min}}{\mu_i\alpha_{\max}} \leq I,\
\frac{\delta\pi g\alpha}{\mu_i\alpha_{\max}} \leq A,\
\frac{\pi}{\alpha_{\max}} \leq I+I_i \leq \frac{\pi}{\alpha_{\min}},
\right. \\
\left.
\frac{\delta\pi}{\mu_i\alpha_{\max}} \leq A+A_i \leq 
\frac{\delta\pi}{\mu\alpha_{\min}}
 \right\}\ .
\end{multline}
Theorem \ref{le1} shows that the compact set ${\cal D}'$ is positively 
invariant 
and attracts all the trajectories.
Therefore, in order to study the asymptotics of system \eqref{eq0}, it is 
sufficient to consider the trajectories of \eqref{eq0} that are in ${\cal D}'$.

\subsection{Equilibria}

 \begin{theo}[Equilibria and asymptotic behavior]
\label{le2}
Define
\begin{gather}
\label{eq2}
\beta
\doteq \frac{\mu}{\alpha_i}-\frac{\mu_i+g}{\alpha} 
\end{gather}
$\bullet$
If $\beta \leq 0$, then there exists a unique equilibrium point of system 
\eqref{eq1} in $\cal D'$, that corresponds to a mite-free situation.
It is globally asymptotically stable, and  given by
\begin{equation}
\label{eq12}
X_{MF} = \frac{\pi}{\alpha}\begin{pmatrix}
1 \\ \frac{\delta}{\mu} \\ 0 \\ 0
\end{pmatrix}\ .
\end{equation}
$\bullet$
If $\beta> \frac{1}{\alpha_i}$, then there exists two 
equilibrium points in $\cal D'$, namely $X_{MF}$ and a coexistence equilibrium 
defined by
\begin{equation}
\label{eq7}
X_{CO} = \frac{\delta\pi g}{\alpha_i(\mu_i+g)}
\frac{\alpha\mu-\alpha_i(\mu_i+g)}{\alpha(\mu+g)-\alpha_i(\mu_i+g)}
\begin{pmatrix}
\frac{1}{\delta}\frac{\alpha_i(\mu_i+g)}{\alpha\mu-\alpha_i(\mu_i+g)} \\
\frac{\alpha}{\alpha\mu-\alpha_i(\mu_i+g)}\\
\frac{\mu_i+g}{\delta g}\\
\frac{1}{g}
\end{pmatrix}\ .
\end{equation}
Moreover, for all initial conditions in $\cal D'$ except in a zero measure set, 
the trajectories tend towards $X_{CO}$.
\end{theo}

Recall that ${\cal R}_0 = \frac{\alpha\mu}{\alpha_i(\mu_i+g)}$, in such a way 
that
\begin{equation}
\beta > 0 \Leftrightarrow {\cal R}_0>1\ .
\end{equation}

The point ${\cal R}_0 = 1$, that is $\beta=0$, is the point of a transcritical 
bifurcation, that appears when ${\cal R}_0$ gets larger than 1.
For larger values, two equilibria are found analytically, a mite-free one, that 
is unstable, and a coexistence equilibrium which is stable.
We've shown (Theorem \ref{le2}) that the latter is globally asymptotically 
stable if $\beta > \frac{1}{\alpha_i}$, and conjecture that the same property 
holds for $\beta$ in the interval $(0, \frac{1}{\alpha_i}]$.
Using $\alpha$ as bifurcation parameter, the bifurcation appears  for 
$\alpha=\frac{\alpha_i(\mu_i +g)}{\mu}\approx 0.125$, 
after substituting the parameter values. 

If we solve numerically the system from \eqref{eq0}, we confirm the existence 
of 
two equilibria when $\alpha$ crosses the bifurcation value of $0.125$. The 
instability and stability of the mite-free and coexistence equilibria, 
respectively is shown in the simulation of figure \ref{fig_infest}.

\begin{figure}
 \centering
 \includegraphics[width=\textwidth]{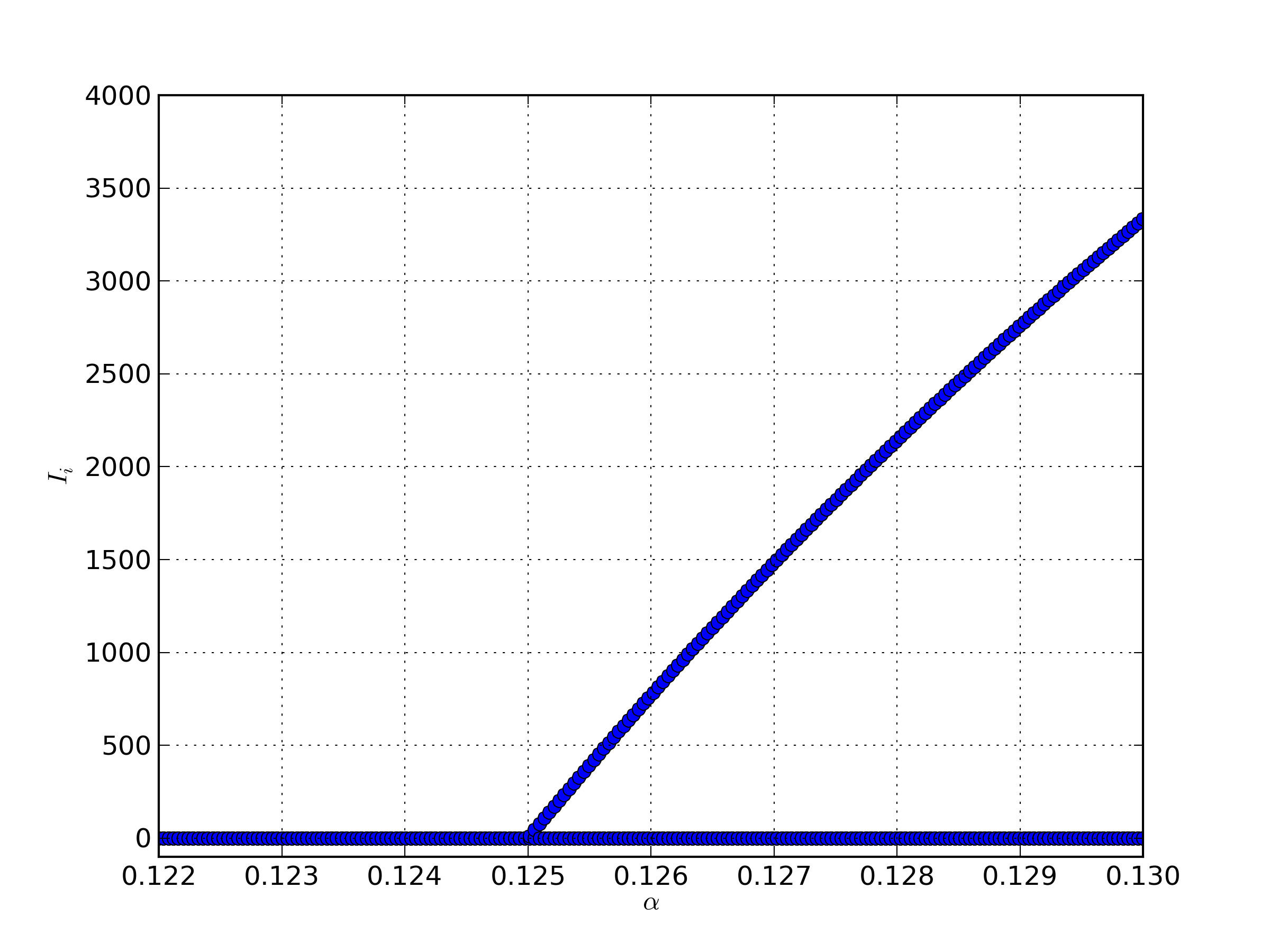}
\label{fig:bifurcation}
\caption{Bifurcation diagram showing the transcritical bifurcation with 
bifurcation point corresponding to $\alpha\approx 0.125$ ($\beta=0$, ${\cal 
R}_0=1$). Blue dots correspond to the equilibrium values of $I_i$}
\end{figure}

 \begin{figure}
  \begin{center}
  \includegraphics[width=\textwidth]{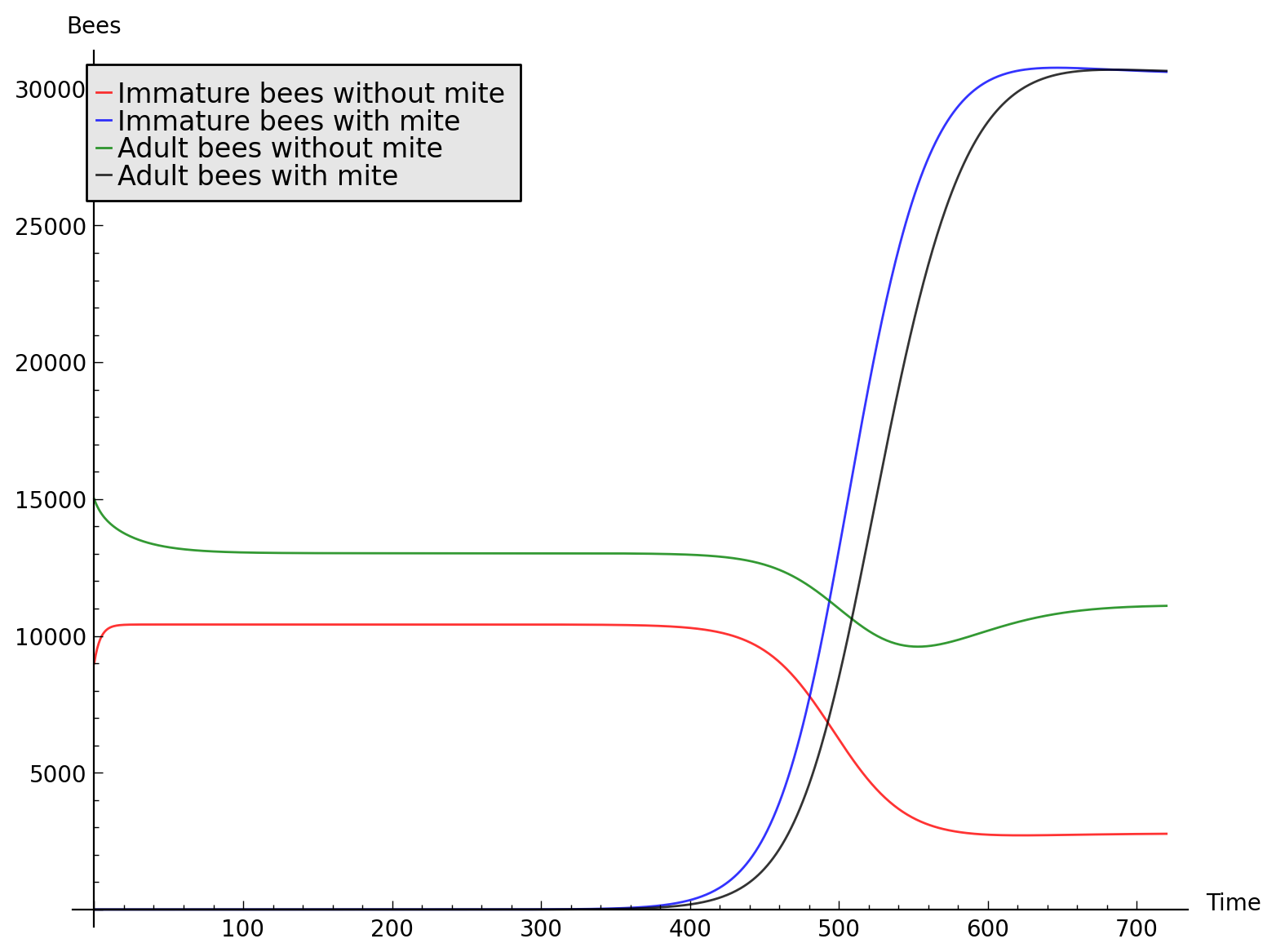} 
  \caption{Simulation showing the infestation of a colony, by a single 
infested adult bee, with parameters giving ${\cal R}_0\approx 1.33$. Initial 
conditions: $I = 5000$, $I_{i}= 0$, $A = 20000$,  $A_{i} = 
0$ and
   parameters $g = 0.01$, $H_i = 0.1$, $\mu=0.04$, $\delta=0.05$, 
$\gamma=10^{-7}$ and $H = 0.19$. On time $t=100$ days, a single infested adult 
bee 
is introduced into the colony. For this simulation, $\beta=0.375$ and ${\cal 
R}_0\approx3.199$}
  \label{fig_infest}
  \end{center}  
   
  \end{figure}
 
 \begin{figure}
  \begin{center}
  \includegraphics[scale=0.5]{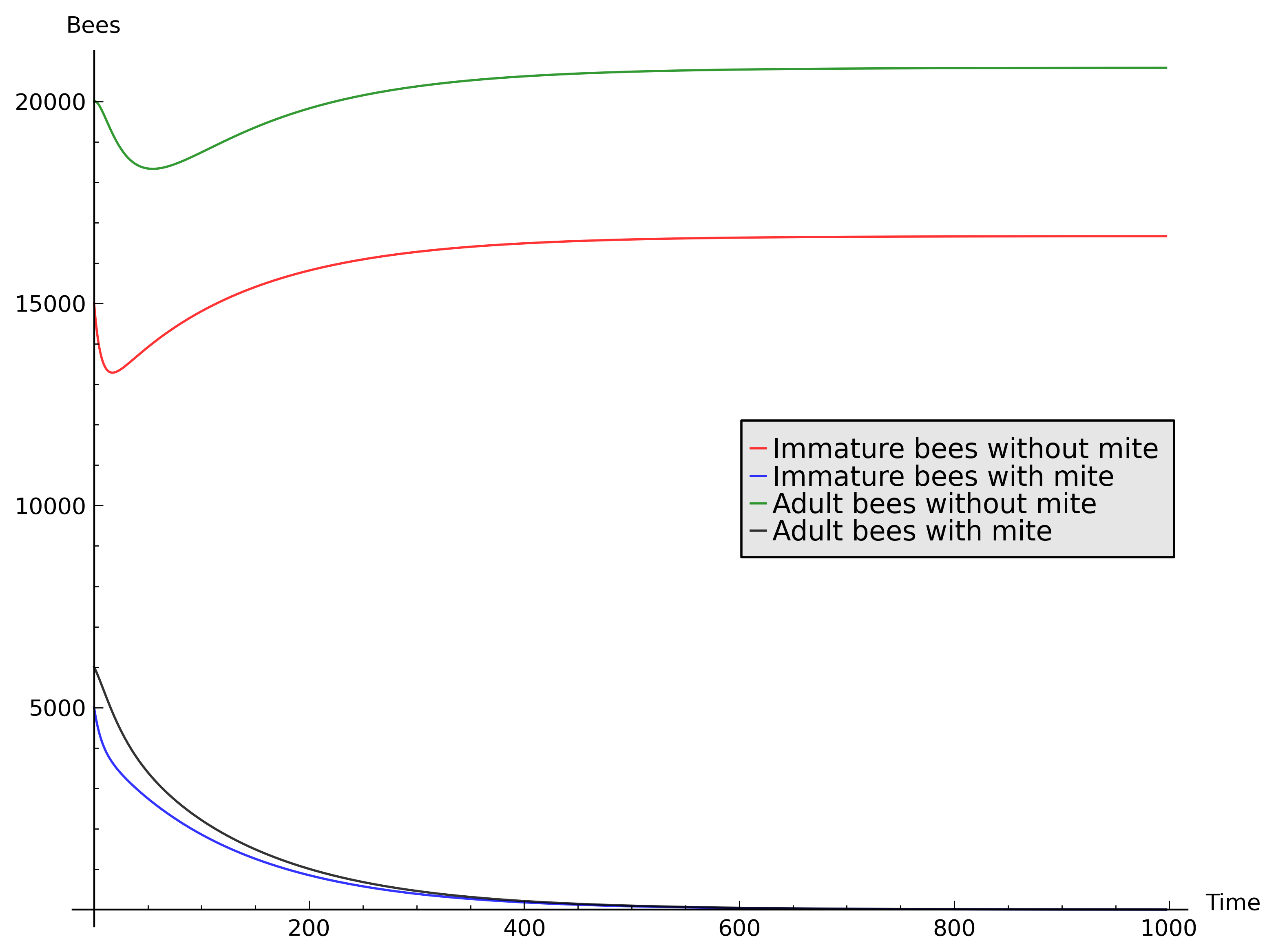} 
  \caption{Simulation showing the elimination of the mites from a colony, by a 
single 
infested adult bee, when $R_0<1$. Initial 
conditions: $I = 15000$, $I_{i}= 5000$, $A = 20000$,  $A_{i} = 
6000$ and
   parameters $g = 0.01$, $H_i = 0.1$, $\mu=0.05$, $\delta=0.05$, 
$\gamma=10^{-7}$ and $H = 0.1$.}
  \label{fig_mitefree}
  \end{center}  
   
  \end{figure}

Figures \ref{fig_infest} and \ref{fig_mitefree} show simulations representing 
the coexistence and mite-free equilibria, respectively.

\section{Discussion and Conclusions} 

Coexistence of bees and Varroa mites in nature is an undeniable fact. However, 
this 
coexistence is fraught with dangers for the bees, since Varroa mites can be 
vectors 
of lethal viral diseases. These deleterious effects for the health of the 
individual workers and the whole colony, has led to the evolution of resistance 
behaviors such as the hygienic behavior and grooming.

Those behaviors are not entirely without cost to the bees, exacerbated hygienic 
behavior -- when both $H$ and $H_i$ are intensified -- can exert a substantial 
toll  on the fitness of the queen. So it is safe to say that this parasitic 
relationship 
has evolved within a vary narrow range of parameters. 
Even if the mite-free equilibrium is advantageous to the colony, maintaining it 
may be too expensive to the bees. 

[we need some discussion regarding the conditions for stability of $X_{MF}$, or 
the invasibility of the colony by mites]

On the other hand, in the absence of viral diseases, mite parasitism seems to 
be 
fairly harmless. If we look at the expression for the ${\cal R}_0$ of 
infestation (\ref{r0}), we can see that the 
mite-induced bee mortality, $\gamma$, (not by viral diseases), 
must be kept low or risk destabilizing the co-existence equilibrium.

Africanized Honey bees, having evolved more effective resistance behaviors, 
are more resistant to CCD by their ability to keep infestation levels lower 
than those of their European 
counterparts\citep{moretto_effects_1991,moretto_heritability_1993}. 
Unfortunately, the lack of more detailed experiments measuring the rates of 
grooming and higienic behaviors in both groups (EHB and AHB), makes it hard to 
position them accurately in the parameter space of the model presented. 

Finally, we hope that the model presented here along with its demonstrated 
dynamical properties will serve as a solid foundation for the development of 
other models including viral dynamics and other aspects of bee colony health.

\section{Appendix -- Proofs of the theorems}

\begin{proof}[Proof of Theorem \ref{le1}]
$\bullet$
Clearly, the right-hand side of the system of equations is globally Lipschitz on 
any subset of $\cal D$ where $A+A_i$ is bounded away from zero.
The existence and uniqueness of the solution of system \eqref{eq0} is then 
obtained for each trajectory staying at finite distance of this boundary.
We will show that the two formulas provided in the statement are valid for each 
trajectory departing initially from a point where $A+A_i\neq 0$.
As a consequence, the fact that all trajectories are defined on infinite horizon 
will ensue.

$\bullet$
Summing up the first two equations in \eqref{eq0} yields, for any point inside 
$\cal D$:
\begin{equation}
\dot I + \dot I_i = \pi - \alpha I - \alpha_iI_i
\geq  \pi - \alpha_{\max}(I+I_i)\ .
\end{equation}
Integrating this differential inequality between any two points $X(0)=X_0$ and 
$X(t)$ of a trajectory for which $X(\tau)\in\cal D$, $\tau\in [0;t]$, one gets
\begin{equation}
\label{eq4}
I(t)+I_i(t) 
\geq \frac{\pi}{\alpha_{\max}} \left(
1-e^{-\alpha_{\max}t}
\right)
+ (I(0)+I_i(0))e^{-\alpha_{\max} t}\ ,
\end{equation}
where the right-hand side is in any case positive for any $t>0$.

Similarly, one has
\begin{equation}
\dot I + \dot I_i \leq  \pi - \alpha_{\min}(I+I_i)\ ,
\end{equation}
and therefore
\begin{equation}
\label{eq5}
I(t)+I_i(t) 
\leq \frac{\pi}{\alpha_{\min}} \left(
1-e^{-\alpha_{\min}t}
\right)
+ (I(0)+I_i(0))e^{-\alpha_{\min} t}\ .
\end{equation}
This proves in particular that the inequalities in \eqref{eq3a} hold for any 
portion of trajectory remaining inside $\cal D$.

We now consider the evolution of $A, A_i$.
Similarly to what was done for $I,I_i$, one has
\begin{equation}
\dot A+\dot A_i = \delta(I+I_i)-\mu A-\mu_iA_i
\geq \delta(I+I_i)-\mu_i(A+A_i)
\end{equation}
Therefore,
\begin{equation}
A(t)+A_i(t)
\geq (A(0)+A_i(0))e^{-\mu_it}
+ \delta \int_0^t (I(\tau)+I_i(\tau)) e^{-\mu_i(t-\tau)}.\ d\tau\ .
\end{equation}
Integrating the lower bound of $I+I_i$ extracted from \eqref{eq4} yields the 
conclusion that any solution departing from $\cal D$ indeed remains in $\cal D$ 
as long as it is defined.
On the other hand, we saw previously that trajectories remaining in $\cal D$ 
could be extended on the whole semi-axis $[0,+\infty)$.
Therefore, any trajectory departing from a point in $\cal D$ can be extended to 
$[0,+\infty)$, and remains in $\cal D$ for any $t>0$.
In particular, \eqref{eq3a} holds for any trajectory departing inside $\cal D$.

Let us now achieve the proof by bounding $A+A_i$ from above.
One has
\begin{equation}
\dot A+\dot A_i \leq \delta(I+I_i)-\mu(A+A_i)
\end{equation}
and thus
\begin{equation}
A(t)+A_i(t)
\leq (A(0)+A_i(0))e^{-\mu t}
+ \delta \int_0^t (I(\tau)+I_i(\tau)) e^{-\mu(t-\tau)}.\ d\tau\ .
\end{equation}
Using \eqref{eq5} then permits to achieve the proof of \eqref{eq3b}, and finally 
the proof of  \eqref{eq3}.

$\bullet$
Let us now prove \eqref{eq03}.
One deduces from \eqref{eq0a} and \eqref{eq0b} and the bounds established 
earlier the differential inequalities
\begin{subequations}
\label{eq25}
\begin{gather}
\dot I \geq \frac{\pi}{\limsup (A+A_i)} A - \alpha I
\geq \frac{\mu\alpha_{\min}}{\delta} A - \alpha I,\\
\dot A \geq \delta I - \mu A +g(\liminf (A+A_i) -A)
\geq \delta I - (\mu+g) A +\frac{\delta\pi g}{\mu_i\alpha_{\max}}
\end{gather}
\end{subequations}
The auxiliary linear time-invariant system
\begin{equation}
\label{eq18}
\frac{d}{dt}
\begin{pmatrix}
I'\\A'
\end{pmatrix}
= \begin{pmatrix}
-\alpha & \frac{\mu\alpha_{\min}}{\delta}\\
\delta & -(\mu+g)
\end{pmatrix}\begin{pmatrix}
I'\\A'
\end{pmatrix}
+ \begin{pmatrix}
0\\\frac{\delta\pi g}{\mu_i\alpha_{\max}}
\end{pmatrix}
\end{equation}
is monotone, as the state matrix involved is a Metzler matrix \citep{?}.
Moreover, it is asymptotically stable, as the associated characteristic 
polynomial is equal to
\begin{equation}
\begin{vmatrix}
s+\alpha & -\frac{\mu\alpha_{\min}}{\delta}\\
-\delta & s+\mu+g
\end{vmatrix}
= s^2+(\alpha+\mu+g)s + \alpha(\mu+g) - \mu\alpha_{\min}\ ,
\end{equation}
and thus Hurwitz because $\alpha(\mu+g) - \mu\alpha_{\min} =  
(\alpha-\alpha_{\min})\mu+\alpha g>0$.
Therefore, all trajectories of \eqref{eq18} tend towards the unique equilibrium:
\begin{eqnarray}
\lim_{t\to +\infty} \begin{pmatrix}
I'(t)\\A'(t)
\end{pmatrix}
& = &
\nonumber
- \begin{pmatrix}
-\alpha & \frac{\mu\alpha_{\min}}{\delta}\\
\delta & -(\mu+g)
\end{pmatrix}^{-1}\begin{pmatrix}
0\\\frac{\delta\pi g}{\mu_i\alpha_{\max}}
\end{pmatrix}\\
& = &
\nonumber
\frac{1}{(\alpha-\alpha_{\min})\mu+\alpha g}
\begin{pmatrix}
\mu+g & \frac{\mu\alpha_{\min}}{\delta}\\
\delta & \alpha
\end{pmatrix}\begin{pmatrix}
0\\\frac{\delta\pi g}{\mu_i\alpha_{\max}}
\end{pmatrix}\\
& = &
\frac{1}{(\alpha-\alpha_{\min})\mu+\alpha g}
\begin{pmatrix}
\frac{\pi g\mu\alpha_{\min}}{\mu_i\alpha_{\max}} \\ \frac{\delta\pi 
g\alpha}{\mu_i\alpha_{\max}}
\end{pmatrix}\ .
\end{eqnarray}
Invoking Kamke's Theorem, see e.g.\ \citep[Theorem 10, p.\ 
29]{coppel1965stability}, one deduces from \eqref{eq25} and the monotony of 
\eqref{eq18} the following comparison result, that holds for all trajectories of 
\eqref{eq6}:
\begin{equation}
\liminf_{t\to +\infty} \begin{pmatrix}
I(t)\\A(t)
\end{pmatrix}
\geq
\frac{1}{(\alpha-\alpha_{\min})\mu+\alpha g}
\begin{pmatrix}
\frac{\pi g\mu\alpha_{\min}}{\mu_i\alpha_{\max}} \\ \frac{\delta\pi 
g\alpha}{\mu_i\alpha_{\max}}
\end{pmatrix}\ .
\end{equation}
This gives \eqref{eq03}.

$\bullet$
One finally proves \eqref{eq03a}.
Using \eqref{eq3b}, identity \eqref{eq0c} implies
\begin{equation}
\dot I_i \geq \frac{\pi}{\limsup (A+A_i)} A_i - \alpha_i I_i
\geq \frac{\mu\alpha_{\min}}{\delta} A_i - \alpha_i I_i 
\end{equation}
Joining this with \eqref{eq0d} and using Kamke's result as before, ones deduces 
that both $I_i$ and $A_i$ have positive values when at least one of their two 
initial values are positive.
This achieves the proof of Theorem \ref{le1}.
\end{proof}

 \begin{proof}[Proof of Theorem \ref{le2}]
The proof is organized as follows.

\begin{enumerate}
\item
We first write system \eqref{eq0} under the form of an I/O system, namely
\begin{subequations}
\label{eq00}
\begin{gather}
\label{eq00a}
\dot I = \pi \frac{A}{A+A_i} - \alpha I\\
\label{eq00b}
\dot A = \delta I -\mu A + u \\
\label{eq00c}
\dot I_i = \pi \frac{A_i}{A+A_i} - \alpha_iI_i\\
\label{eq00d}
\dot A_i = \delta I_i -(\mu_i+g)A_i\\
\label{eq00e}
y = gA_i
\end{gather}
where $u$, resp.\ $y$, is the input, resp.\ the output, closed by the unitary 
feedback
\begin{equation}
\label{eq00f}
u = y\ .
\end{equation}
\end{subequations}
For subsequent use of the theory of monotone systems, one determines, for any 
(nonnegative) constant value of $u$, the equilibrium values of $(I, A, I_i, 
A_i)$ for equation \eqref{eq00a}-\eqref{eq00d}, and the corresponding values of 
$y$ as given by \eqref{eq00e}.

\item
The equilibrium points of system \eqref{eq0} are then exactly (and easily) 
obtained by solving the fixed point problem $u=y$ among the solutions of the 
previous problem.

unique equilibrium points when $\beta \leq 0$, and there exist exactly two 
equilibrium points when $\beta >0$.
points.

\item
One then shows that the I/O system $u\mapsto y$ defined by 
\eqref{eq00a}-\eqref{eq00e} is anti-monotone with respect to certain order 
relation, and the study of the stability of these equilibria shows that it 
admits single-valued I/S and I/O characteristics, as in 
\citep{angeli2004interconnections}.

\item
Using this properties, the stability of the equilibria of the system obtained by 
closing the loop  \eqref{eq00a}-\eqref{eq00e} by \eqref{eq00f} is then 
established using arguments similar to \citet{angeli2003monotone}.

\end{enumerate}

\noindent
1.\
For fixed $u>0$, the equilibrium equations of the I/O system \eqref{eq00} are 
given by
\begin{subequations}
\label{eq6}
\begin{gather}
\label{eq6a}
\pi \frac{A}{A+A_i} - \alpha I = 0\\
\label{eq6c}
\delta I -\mu A + u  = 0 \\
\label{eq6b}
\pi \frac{A_i}{A+A_i} - \alpha_iI_i  = 0\\
\label{eq6d}
\delta I_i -(\mu_i+g)A_i  = 0\\
y = gA_i
\end{gather}
\end{subequations}
Summing up the first and third identities gives
\begin{equation}
\pi = \alpha I + \alpha_iI_i\ ,
\end{equation}
and thus necessarily:
\begin{equation}
\label{eq8}
\exists \lambda\in [0;1],\qquad
I = \lambda\frac{\pi}{\alpha},\qquad
I_i = (1-\lambda)\frac{\pi}{\alpha_i}\ .
\end{equation}

$\bullet$ The case $\lambda =0$ yields $I=0$, and then $A=0$ by \eqref{eq6a}, 
and therefore $u$ has to be zero from \eqref{eq6c}.
Also, $I_i=\frac{\pi}{\alpha_i}$, $A_i=\frac{\delta\pi}{\alpha_i(\mu_i+g)}$ by 
\eqref{eq6d}, and then $y=gA_i= \frac{g\delta\pi}{\alpha_i(\mu_i+g)}$.
in \eqref{eq19} and should be discarded.
obtained point is located outside $\cal D$ and has to be discarded; or 

$\bullet$ The case $\lambda =1$ yields $I_i=0$, and then $A_i=0$ by \eqref{eq6d} 
or \eqref{eq6b}, and $y=0$.
There remains the two following conditions:
\begin{equation}
\label{eq11}
\pi = \alpha I,\qquad
\delta I = \mu A - u
\end{equation}
which yield
\begin{equation}
\label{eq11a}
I = \frac{\pi}{\alpha},\qquad
A = \frac{\delta\pi}{\alpha\mu} + \frac{u}{\mu}
\end{equation}
unconditionally.

$\bullet$ Let us now look for possible values of $\lambda$ in $(0;1)$.
From \eqref{eq8} and \eqref{eq6a}-\eqref{eq6b}, one deduces
\begin{equation}
\label{eq9}
\frac{A}{A_i}
= \frac{\alpha I}{\alpha_iI_i}
= \frac{\lambda}{1-\lambda}\ .
\end{equation}
Using \eqref{eq8} on the one hand and summing the two identities 
\eqref{eq6c}-\eqref{eq6d} on the other hand, yields
\begin{equation}
\delta (I+I_i)
= \delta \pi \left(
\frac{\lambda}{\alpha} + \frac{1-\lambda}{\alpha_i}
\right)
= \mu A+ (\mu_i+g)A_i - u
= A \left(
\mu + (\mu_i+g)\frac{1-\lambda}{\lambda}
\right) - u\ .
\end{equation}
This permits to express $A$ as a function of $\lambda$, namely:
\begin{equation}
A = \frac{\lambda}{\lambda\mu + (1-\lambda)(\mu_i+g)} \left[
\delta\pi\left(
\frac{\lambda}{\alpha} + \frac{1-\lambda}{\alpha_i}
\right)
+ u \right]\ .
\end{equation}
Using this formula together with \eqref{eq8}, \eqref{eq6d} and \eqref{eq9} now 
allows to find an equation involving only the unknown $\lambda$, namely:
\begin{multline}
\delta I_i
= \frac{\delta\pi}{\alpha_i}(1-\lambda)
= (\mu_i+g)A_i
= (\mu_i+g)\frac{A_i}{A}A\\
= (\mu_i+g) \frac{1-\lambda}{\lambda} \frac{\lambda}{\lambda\mu + 
(1-\lambda)(\mu_i+g)} \left[
\delta\pi\left(
\frac{\lambda}{\alpha} + \frac{1-\lambda}{\alpha_i}
\right)
+ u \right]\ .
\end{multline}
Simplifying (as $\lambda \neq 0, 1$) gives:
\begin{equation}
\frac{\delta\pi}{\alpha_i}
= \frac{\mu_i+g}{\lambda\mu + (1-\lambda)(\mu_i+g)} \left[
\delta\pi\left(
\frac{\lambda}{\alpha} + \frac{1-\lambda}{\alpha_i}
\right)
+ u \right]\ .
\end{equation}
The previous condition is clearly affine in $\lambda$.
It writes
\begin{equation}
\left(
\lambda\mu + (1-\lambda)(\mu_i+g)
\right)\frac{\delta\pi}{\alpha_i}
= (\mu_i+g) \left(
\delta\pi\left(
\frac{\lambda}{\alpha} + \frac{1-\lambda}{\alpha_i}
\right)
+ u \right)
\end{equation}
which, after developing and simplifying, can be expressed as:
\begin{equation}
\lambda\mu\frac{\delta\pi}{\alpha_i}
= (\mu_i+g) \left(
\delta\pi\frac{\lambda}{\alpha} + u
\right)
\end{equation}
and finally
\begin{equation}
(\mu_i+g) u
= \delta\pi\left(
\frac{\mu}{\alpha_i} - \frac{\mu_i+g}{\alpha}
\right)\lambda
= \delta\pi\beta\lambda\ .
\end{equation}
For $u\geq 0$, this equation admits a solution in $(0;1)$ if and only if
\begin{equation}
\label{eq21a}
\beta >0 \qquad\text{ and } \qquad u < u^* \doteq 
\frac{\delta\pi\beta}{\mu_i+g}\ ,
\end{equation}
and the latter is given as
\begin{equation}
\label{eq21}
\lambda = \frac{\mu_i+g}{\delta\pi\beta} u\ .
\end{equation}
The state and output values may then be expressed explicitly as functions of 
$u$.
In particular, one has
\begin{equation}
\label{eq20}
y(u) = gA_i = \frac{\delta g}{\mu_i+g} I_i
= \frac{\delta\pi g}{\alpha_i(\mu_i+g)} (1-\lambda)
= \frac{\delta\pi g}{\alpha_i(\mu_i+g)} \left(
1-\frac{\mu_i+g}{\delta\pi\beta} u
\right)\ .
\end{equation}

$\bullet$
\eqref{eq6} admits exactly one solution in ${\cal D}'$ for any $u \geq 0$;
admits a supplementary solution in ${\cal D}'$ for any $u \in [0;u^*)$.
The following tables summarize the number of solutions of \eqref{eq6} for all 
nonnegative values of $u$.
\begin{figure}[h]
\begin{center}
\begin{tabular}{|c|c|}
\hline
Values of $u\geq 0$ & Number of distinct solutions of \eqref{eq6}\\
\hline
$u=0$ & 2 \\
$0<u$ & 1\\
\hline
\end{tabular}
\caption{${\cal R}_0\leq 1$ (i.e.\ $\beta \leq 0$).}
\end{center}
\end{figure}
\begin{figure}[h]
\begin{center}
\begin{tabular}{|c|c|}
\hline
Values of $u\geq 0$ & Number of distinct solutions of \eqref{eq6}\\
\hline
$u=0$ & 3 \\
$0 < u < u^*$ & 2 \\
$u^* \leq u$ & 1\\
\hline
\end{tabular}
\caption{${\cal R}_0> 1$ (i.e.\ $\beta > 0$).}
\end{center}
\end{figure}

2.\
The equilibrium points of system \eqref{eq0} are exactly those points for which 
$u=y(u)$ for some nonnegative scalar $u$, where $y(u)$ is one of the output 
values corresponding to $u$ previously computed.
We now examine in more details the solutions of this equation.

$\bullet$
For the value $\lambda =0$ in the previous computations, one should have $u=0$, 
due to \eqref{eq21}; but on the other hand $y >0$ for $u=0$, due to 
\eqref{eq20}.
Therefore this point does not correspond to an equilibrium point of system 
\eqref{eq6}.

$\bullet$
The value $\lambda = 1$ yields a unique equilibrium point.
Indeed, $y=0$, so $u$ should be zero too, and the unique solution is given by
\begin{equation}
I = \frac{\pi}{\alpha},\ A = \frac{\delta\pi}{\alpha\mu},\
I_i = 0,\ A_i = 0,\ y=0\ .
\end{equation}
This corresponds to the equilibrium denoted $X_{MF}$ in the statement.

$\bullet$
Let us consider now the case of $\lambda\in (0;1)$.
For this case to be considered, it is necessary that $\beta>0$, that is ${\cal 
R}_0>1$.
The value of $u$ should be such that (see \eqref{eq20})
\begin{equation}
y =  \frac{\delta\pi g}{\alpha_i(\mu_i+g)} -\frac{g}{\alpha_i\beta} u = u\ ,
\end{equation}
that is
\begin{equation}
\left(
1 + \frac{g}{\alpha_i\beta}
\right) u
=  \frac{\delta\pi g}{\alpha_i(\mu_i+g)}\ ,
\end{equation}
or again
\begin{equation}
\label{eq23}
u = \frac{\delta\pi \beta g}{(\alpha_i\beta+g)(\mu_i+g)}
= \frac{\delta\pi g}{\alpha_i(\mu_i+g)}
\frac{\alpha\mu-\alpha_i(\mu_i+g)}{\alpha(\mu+g)-\alpha_i(\mu_i+g)} \ ,
\end{equation}
after replacing $\beta$ by its value defined in \eqref{eq2}.
The corresponding value of
\begin{equation}
\lambda = \frac{\mu_i+g}{\delta\pi\beta} u
= \frac{g}{\alpha_i\beta+g}\ ,
\end{equation}
given by \eqref{eq21}, is clearly contained in $(0;1)$ when $\beta>0$.
Therefore, when $\beta>0$, there also exists a second equilibrium.
The latter is given by:

\begin{subequations}
\label{eq26}
\begin{gather}
\label{eq26a}
I = \lambda\frac{\pi}{\alpha}
= \frac{\mu_i+g}{\alpha\delta\beta} u
= \frac{1}{\delta}\frac{\alpha_i(\mu_i+g)}{\alpha\mu-\alpha_i(\mu_i+g)} u,\qquad
A_i = \frac{u}{g},\\
\label{eq26b}
I_i = \frac{\mu_i+g}{\delta}A_i
= \frac{\mu_i+g}{\delta g}u\\
\label{eq26ba}
A = \frac{1}{\mu} \left(
\delta I+u
\right)
= \frac{1}{\mu} \left(
\frac{\alpha_i(\mu_i+g)}{\alpha\mu-\alpha_i(\mu_i+g)} + 1
\right)u
= \frac{\alpha}{\alpha\mu-\alpha_i(\mu_i+g)}u\ ,
\end{gather}
\end{subequations}
and corresponds to $X_{CO}$ defined in the statement.


diagonal  that comes from the loop closing.

3.\
Let $\cal K$ be the cone in $\Rset_+^4$ defined as the product of orthants 
$\Rset_+\times\Rset_+\times\Rset_-\times\Rset_-$.
We endow the state space with this order.
In other words, for any $X=(I, A, I_i, A_i)$ and $X'=(I', A', I'_i, A'_i)$ in 
$\Rset_+^4$, $X \leq_{\cal K} X'$ means:
\begin{equation}
I\leq I', A\leq A', I_i\geq I'_i, A_i\geq A'_i\ .
\end{equation}
With this structure, one may verify that the system \eqref{eq00a}-\eqref{eq00e} 
has the following monotonicity properties 
\citep{hirsch1988stability,smith2008monotone}
\begin{itemize}
\item
For any  function $u\in{\cal U}\doteq \{ u: [0;+\infty)\to\Rset,$ locally 
integrable and taking on positive values almost everywhere$\}$, for any 
$X_0,X'_0 \in \Rset_+^4$,
\begin{equation}
X_0 \leq_{\cal K} X'_0 \quad \Rightarrow \quad
\forall t\geq 0,\
X(t; X_0, u) \leq _{\cal K} X(t; X'_0, u)
\end{equation}
where by definition $X(t; X_0, u)$ denotes the value at time $t$ of the point in 
the trajectory departing at time $0$ from $X_0$ and subject to input $u$.

\item
The Jacobian matrix of the I/O system is
\begin{equation}
\label{eq15a}
\begin{pmatrix}
-\alpha & \pi\frac{A_i}{(A+A_i)^2} & 0 & -\pi\frac{A}{(A+A_i)^2}\\
\delta & -\mu & 0 & 0\\
0 & -\pi\frac{A_i}{(A+A_i)^2} & -\alpha_i & \pi\frac{A}{(A+A_i)^2}\\
0 & 0 & \delta & -(\mu_i+g)
\end{pmatrix}\ ,
\end{equation}
which is irreducible when $A\neq 0$ and $A_i\neq 0$.
The system is therefore strongly monotone in ${\cal D}'\setminus\{X\ :\ A_i=0\}$ 
(notice that ${\cal D}'$ does not contain points for which $A=0$), and also on 
the invariant subset ${\cal D}'\cap\{X\ :\ I_i = 0,\ A_i=0,\}$.

\item
The input-to-state map is monotone, that is: for any inputs $u,u'\in\cal U$, for 
any $X_0\in\Rset_+^4$,
\begin{equation}
u(t) \leq u'(t) \text{ a.e.} \quad \Rightarrow \quad
\forall t\geq 0,\
X(t; X_0, u) \leq _{\cal K} X(t; X'_0, u)\ .
\end{equation}

\item
The state-to-output map is anti-monotone, that is:
for any $X,X' \in \Rset_+^4$,
\begin{equation}
X \leq_{\cal K} X' \quad \Rightarrow \quad
\forall t\geq 0,\
gA_i \geq gA'_i\ 
\end{equation}
\end{itemize}

monotone (due to the irreducibility of the Jacobian matrix) for any constant 
value of $u$.

\comment{
Moreover, one shows that
\begin{equation}
\label{eq27}
X_{CO} \ll_{\cal K} X_{MF}\ .
\end{equation}
Comparing the expressions given in the statement, it is indeed sufficient, in 
order to show this, to prove that
\begin{subequations}
\begin{gather}
\frac{\pi}{\alpha}
> \frac{\delta\pi g}{\alpha_i(\mu_i+g)}
\frac{\alpha\mu-\alpha_i(\mu_i+g)}{\alpha(\mu+g)-\alpha_i(\mu_i+g)}
\frac{1}{\delta}\frac{\alpha_i(\mu_i+g)}{\alpha\mu-\alpha_i(\mu_i+g)}\\
\frac{\pi}{\alpha} \frac{\delta}{\mu}
> \frac{\delta\pi g}{\alpha_i(\mu_i+g)}
\frac{\alpha\mu-\alpha_i(\mu_i+g)}{\alpha(\mu+g)-\alpha_i(\mu_i+g)}
\frac{\alpha}{\alpha\mu-\alpha_i(\mu_i+g)}
\end{gather}
\end{subequations}
that is:
\begin{subequations}
\begin{gather}
1 > \frac{\alpha g}{\alpha(\mu+g)-\alpha_i(\mu_i+g)}\\
\frac{\alpha_i(\mu_i+g)}{\alpha\mu}
> \frac{\alpha g}{\alpha(\mu+g)-\alpha_i(\mu_i+g)}
\end{gather}
\end{subequations}
Introducing now $\beta$ defined in \eqref{eq2}, this is equivalent to
\begin{subequations}
\begin{gather}
1 > \frac{\alpha g}{\alpha g +\alpha\alpha_i\beta}\\
\frac{\alpha_i(\mu_i+g)}{\alpha\mu}
> \frac{\alpha g}{\alpha g +\alpha\alpha_i\beta}
\end{gather}
\end{subequations}
}

$\bullet$
In order to construct I/S and I/O characteristics for system \eqref{eq6}, we now 
examine the stability of the equilibria of system \eqref{eq6} for any fixed 
value of $u\in\Rset_+$.
As shown by Theorem \ref{le1}, all trajectories are precompact.

$\bullet$
When $\beta\leq 0$, it has been previously established that for any $u\in\Rset$ 
there exists at most one equilibrium in ${\cal D}'$ to the I/O system 
\eqref{eq6}.
The strong monotonicity property of this system depicted above then implies that 
this equilibrium is globally attractive \citep[Theorem 
10.3]{hirsch1988stability}.
Therefore, system \eqref{eq6} possesses I/S and I/O characteristics.
As for any value of $u$, this equilibrium corresponds to zero output, the I/O 
characteristics is zero.
Applying the results of \citet{angeli2004interconnections}, one gets that the 
closed-loop system 
equilibrium $X_{MF}$ is an almost globally attracting equilibrium for system 
\eqref{eq0}.

$\bullet$
Let us now consider the case where $\beta> 0$.
We first show that the equilibrium point with $I_i=0, A_i=0$ and \eqref{eq11} is 
locally unstable.
Notice that this point is located on a branch of solution parametrized by $u$ 
and departing from $X_{MF}$ for $u=0$.
The Jacobian matrix \eqref{eq15a} taken at this point is
\begin{equation}
\begin{pmatrix}
-\alpha & 0 & 0 & -\frac{\mu\alpha\pi}{\delta\pi + \alpha u}\\
\delta & -\mu & 0 & 0\\
0 & 0 & -\alpha_i & \frac{\mu\alpha\pi}{\delta\pi + \alpha u}\\
0 & 0 & \delta & -(\mu_i+g)
\end{pmatrix}\ .
\end{equation}
This matrix is block triangular, with diagonal blocks
\begin{equation}
\begin{pmatrix}
-\alpha & 0\\
\delta & -\mu
\end{pmatrix}
\qquad \text{ and } \qquad
\begin{pmatrix}
-\alpha_i & \frac{\mu\alpha\pi}{\delta\pi + \alpha u}\\
\delta & -(\mu_i+g)
\end{pmatrix}\ .
\end{equation}
The first of them is clearly Hurwitz, while the second, whose characteristic 
polynomial is 
\begin{multline}
s^2+(\alpha_i+\mu_i+g)s +\alpha_i(\mu_i+g)-\frac{\mu\alpha\delta\pi}{\delta\pi + 
\alpha u}
= s^2+(\alpha_i+\mu_i+g)s -\alpha\alpha_i(\beta-u(\mu_i+g))\\
= s^2+(\alpha_i+\mu_i+g)s -\alpha\alpha_i(\mu_i+g)(u^*-u)
\end{multline}
(where $u^*$ is defined in \eqref{eq21a}) is not Hurwitz when $\beta>0$ and 
$0\leq u\leq u^*$, and has a positive root for $0<u<u^*$.
Therefore, the corresponding equilibrium of the I/O system \eqref{eq00} is 
unstable for these values of $u$.

The other solution, given as a function of $u$ by \eqref{eq26}, is located on a 
branch of solution parametrized by $u$ and departing from $X_{CO}$ for $u=0$.
As the other solution is unstable for $0 < u < u^*$, one can deduce from 
\citet[Theorem 10.3] {hirsch1988stability} that these solutions are locally 
asymptotically stable.

$\bullet$
One may now associate to any $u\in [0;u^*]$ the corresponding unique locally 
asymptotically stable equilibrium point, and the corresponding output value, 
defining therefore respectively an I/S characteristic $k_X$ and an I/O 
characteristic $k$ for system \eqref{eq00}.

For any scalar $u\in [0;u^*]$, for almost any $X_0\in{\cal D}'$, one has
\begin{equation}
\lim_{t\to +\infty} X(t;X_0,u) = k_X(u),\qquad
\lim_{t\to +\infty} y(t;X_0,u) = k(u)\ ,
\end{equation}
and, from the monotony properties, for any scalar-valued continuous function 
$u$, for almost any $X_0\in{\cal D}'$:
\begin{equation}
k\left(
\limsup_{t\to +\infty} u(t)
\right)
\leq \liminf_{t\to +\infty} y(t;X_0,u)
\leq \limsup_{t\to +\infty} y(t;X_0,u)
\leq k\left(
\liminf_{t\to +\infty} u(t)
\right)\ .
\end{equation}
Using the fact that $k$ is anti-monotone and that $u=y$ for the closed-loop 
system, one deduces, as e.g.\ in \citet{gouze1988criterion} that, for the 
solutions of the latter,
\begin{multline}
\label{eq40}
k^{2l} \left(
\liminf_{t\to +\infty} y(t;X_0,u)
\right)
\leq \liminf_{t\to +\infty} y(t;X_0,u)
\leq \limsup_{t\to +\infty} y(t;X_0,u)
\leq k^{2l}\left(
\limsup_{t\to +\infty} y(t;X_0,u)
\right)\ .
\end{multline}

Here $k(u)$, defined by \eqref{eq20}, is a linear decreasing map.
When its slope is smaller than 1, then the sequences in the left and right of 
\eqref{eq40} tend towards the fixed point that corresponds to the output value 
at $X=X_{CO}$, see \eqref{eq23}.

This slope value, see \eqref{eq20}, is equal to
\begin{equation}
\frac{\delta\pi g}{\alpha_i(\mu_i+g)} \frac{\mu_i+g}{\delta\pi\beta}
= \frac{1}{\alpha_i\beta}\ ,
\end{equation}
and it thus smaller than 1 if and only if $\beta> \frac{1}{\alpha_i}$, which is 
an hypothesis of the statement.

Under these assumptions, one then obtains that the $\liminf$ and $\limsup$ in 
\eqref{eq40} are equal, and thus that $y$, and thus $u$, possesses limit for 
$t\to +\infty$.
Moreover, the state itself converges towards the equilibrium $X_{CO}$ when $t\to 
+\infty$ for almost every initial conditions $X(0)$.
This achieves the proof of Theorem \ref{le2}.
\end{proof}

\begin{acknowledgements}
The authors would like to thank Funda\c{c}\~ao Getulio Vargas for financial 
support 
in the form of a scholarship to Joyce  de Figueir\'o Santos. They are also 
grateful for valuable comments by Moacyr A. H. Silva, Max O. Souza and Jair 
Koiler on an early version of the manuscript.
\end{acknowledgements}

\section*{Compliance with Ethical Standards}
Conflict of Interest: The authors declare that they have no conflict of 
interest.

\bibliographystyle{spbasic}      

\begin{thebibliography}{23}
\providecommand{\natexlab}[1]{#1}
\providecommand{\url}[1]{{#1}}
\providecommand{\urlprefix}{URL }
\expandafter\ifx\csname urlstyle\endcsname\relax
  \providecommand{\doi}[1]{DOI~\discretionary{}{}{}#1}\else
  \providecommand{\doi}{DOI~\discretionary{}{}{}\begingroup
  \urlstyle{rm}\Url}\fi
\providecommand{\eprint}[2][]{\url{#2}}

\bibitem[{Angeli and Sontag(2004)}]{angeli2004interconnections}
Angeli D, Sontag E (2004) Interconnections of monotone systems with
  steady-state characteristics. In: Optimal control, stabilization and
  nonsmooth analysis, Springer, pp 135--154

\bibitem[{Angeli and Sontag(2003)}]{angeli2003monotone}
Angeli D, Sontag ED (2003) Monotone control systems. Automatic Control, IEEE
  Transactions on 48(10):1684--1698

\bibitem[{Arechavaleta-Velasco and
  Guzman-Novoa(2001)}]{arechavaleta-velasco_relative_2001}
Arechavaleta-Velasco ME, Guzman-Novoa E (2001) Relative effect of four
  characteristics that restrain the population growth of the mite
  \textit{Varroa destructor} in honey bee ( \textit{Apis mellifera} ) colonies.
  Apidologie 32(2):157--174, \doi{10.1051/apido:2001121},
  
\urlprefix\url{
http://www.apidologie.org/index.php?option=com_article&access=standard&Itemid=12
9&url=/articles/apido/pdf/2001/02/velasco.pdf},
  00000

\bibitem[{Calderón et~al(2010)Calderón, Veen, Sommeijer, and
  Sanchez}]{calderon_reproductive_2010}
Calderón RA, Veen JWv, Sommeijer MJ, Sanchez LA (2010) Reproductive biology of
  varroa destructor in africanized honey bees (apis mellifera). Experimental
  and Applied Acarology 50(4):281--297, \doi{10.1007/s10493-009-9325-4},
  \urlprefix\url{http://link.springer.com/article/10.1007/s10493-009-9325-4},
  00010

\bibitem[{Carneiro et~al(2007)Carneiro, Torres, Strapazzon, Ramírez,
  Guerra~Jr, Koling, and Moretto}]{carneiro_changes_2007}
Carneiro FE, Torres RR, Strapazzon R, Ramírez SA, Guerra~Jr JCV, Koling DF,
  Moretto G (2007) Changes in the reproductive ability of the mite varroa
  destructor (anderson e trueman) in africanized honey bees (apis mellifera l.)
  (hymenoptera: Apidae) colonies in southern brazil. Neotropical Entomology
  36(6):949--952, \doi{10.1590/S1519-566X2007000600018},
  
\urlprefix\url{
http://www.scielo.br/scielo.php?pid=S1519-566X2007000600018&script=sci_arttext},
  00022

\bibitem[{Coppel(1965)}]{coppel1965stability}
Coppel WA (1965) Stability and asymptotic behavior of differential equations,
  vol~11. Heath Boston

\bibitem[{Corrêa-Marques et~al(1998)Corrêa-Marques, David, and
  {others}}]{correa-marques_uncapping_1998}
Corrêa-Marques MH, David DE, {others} (1998) Uncapping of worker bee brood, a
  component of the hygienic behavior of africanized honey bees against the mite
  varroa jacobsoni oudemans. Apidologie 29(3):283--289,
  
\urlprefix\url{
http://hal.archives-ouvertes.fr/docs/00/89/14/94/PDF/hal-00891494.pdf}

\bibitem[{Van~den Driessche and
  Watmough(2002)}]{van_den_driessche_reproduction_2002}
Van~den Driessche P, Watmough J (2002) Reproduction numbers and sub-threshold
  endemic equilibria for compartmental models of disease transmission.
  Mathematical biosciences 180(1):29--48,
  
\urlprefix\url{
http://www.sciencedirect.com/science/article/pii/S0025556402001086}

\bibitem[{Gouz{\'e}(1988)}]{gouze1988criterion}
Gouz{\'e} JL (1988) {A criterion of global convergence to equilibrium for
  differential systems. Application to Lotka-Volterra systems}. Research Report
  RR-0894, \urlprefix\url{https://hal.inria.fr/inria-00075661}

\bibitem[{Hirsch(1988)}]{hirsch1988stability}
Hirsch MW (1988) Stability and convergence in strongly monotone dynamical
  systems. J reine angew Math 383(1):53

\bibitem[{Khoury et~al(2011)Khoury, Myerscough, and
  Barron}]{khoury_quantitative_2011}
Khoury DS, Myerscough MR, Barron AB (2011) A quantitative model of honey bee
  colony population dynamics. {PLoS} {ONE} 6(4):e18,491,
  \doi{10.1371/journal.pone.0018491},
  \urlprefix\url{http://dx.doi.org/10.1371/journal.pone.0018491}

\bibitem[{Medina and Martin(1999)}]{medina_comparative_1999}
Medina LM, Martin SJ (1999) A comparative study of varroa jacobsoni
  reproduction in worker cells of honey bees (apis mellifera) in england and
  africanized bees in yucatan, mexico. Experimental \& Applied Acarology
  23(8):659--667, \doi{10.1023/A:1006275525463},
  \urlprefix\url{http://link.springer.com/article/10.1023/A%3A1006275525463}

\bibitem[{Mondragón et~al(2005)Mondragón, Spivak, and
  Vandame}]{mondragon_multifactorial_2005}
Mondragón L, Spivak M, Vandame R (2005) A multifactorial study of the
  resistance of honeybees \textit{Apis mellifera} to the mite \textit{Varroa
  destructor} over one year in mexico. Apidologie 36(3):345--358,
  \doi{10.1051/apido:2005022},
  
\urlprefix\url{
http://www.apidologie.org/index.php?option=com_article&access=standard&Itemid=12
9&url=/articles/apido/pdf/2005/03/M4080.pdf},
  00000

\bibitem[{Moretto et~al(1991)Moretto, Gonçalves, De~Jong, Bichuette, and
  {others}}]{moretto_effects_1991}
Moretto G, Gonçalves LS, De~Jong D, Bichuette MZ, {others} (1991) The effects
  of climate and bee race on varroa jacobsoni oud infestations in brazil.
  Apidologie 22(3):197--203,
  
\urlprefix\url{
http://hal.archives-ouvertes.fr/docs/00/89/09/07/PDF/hal-00890907.pdf}

\bibitem[{Moretto et~al(1993)Moretto, Gonçalves, and
  De~Jong}]{moretto_heritability_1993}
Moretto G, Gonçalves LS, De~Jong D (1993) Heritability of africanized and
  european honey bee defensive behavior against the mite varroa jacobsoni.
  Revista Brasileira de Genetica 16:71--71

\bibitem[{Oldroyd(2007)}]{oldroyd_whats_2007}
Oldroyd BP (2007) What's killing american honey bees? {PLoS} Biol 5(6):e168,
  \doi{10.1371/journal.pbio.0050168},
  \urlprefix\url{http://dx.doi.org/10.1371/journal.pbio.0050168}, 00205

\bibitem[{Pereira et~al(2002)Pereira, Lopes, Camargo, and Vilela}]{Embrapa}
Pereira FdM, Lopes MTR, Camargo RCR, Vilela SLO (2002) Organização social e
  desenvolvimento das abelhas apis mellifera.
  
\urlprefix\url{
http://sistemasdeproducao.cnptia.embrapa.br/FontesHTML/Mel/SPMel/organizacao.htm
}

\bibitem[{Pinto et~al(2012)Pinto, Puker, Barreto, and
  Message}]{pinto_ectoparasite_2012}
Pinto FA, Puker A, Barreto LMRC, Message D (2012) The ectoparasite mite varroa
  destructor anderson and trueman in southeastern brazil apiaries: effects of
  the hygienic behavior of africanized honey bees on infestation rates. Arquivo
  Brasileiro de Medicina Veterinária e Zootecnia 64(5):1194--1199,
  \doi{10.1590/S0102-09352012000500017},
  
\urlprefix\url{
http://www.scielo.br/scielo.php?script=sci_abstract&pid=S0102-09352012000500017&
lng=en&nrm=iso&tlng=en}

\bibitem[{Ratti et~al(2012)Ratti, Kevan, and Eberl}]{ratti_mathematical_2012}
Ratti V, Kevan PG, Eberl HJ (2012) A mathematical model for population dynamics
  in honeybee colonies infested with varroa destructor and the acute bee
  paralysis virus. Canadian Applied Mathematics Quarterly

\bibitem[{Smith(2008)}]{smith2008monotone}
Smith HL (2008) Monotone dynamical systems: an introduction to the theory of
  competitive and cooperative systems, vol~41. American Mathematical Soc.

\bibitem[{Spivak(1996)}]{spivak_honey_1996}
Spivak M (1996) Honey bee hygienic behavior and defense against varroa
  jacobsoni. Apidologie 27:245--260,
  
\urlprefix\url{
http://www.apidologie.org/index.php?option=com_article&access=standard&Itemid=12
  9&url=/articles/apido/pdf/1996/04/Apidologie_0044-8435_1996_27_4_ART0007.pdf}

\bibitem[{Vandame et~al(2000)Vandame, Colin, Morand, and
  Otero-Colina}]{vandame_levels_2000}
Vandame R, Colin ME, Morand S, Otero-Colina G (2000) Levels of compatibility in
  a new host-parasite association: Apis {mellifera/Varroa} jacobsoni. Canadian
  Journal of Zoology 78(11):2037--2044, \doi{10.1139/z00-109},
  \urlprefix\url{http://www.nrcresearchpress.com/doi/abs/10.1139/z00-109},
  00023

\bibitem[{Vandame et~al(2002)Vandame, Morand, Colin, and
  Belzunces}]{vandame_parasitism_2002}
Vandame R, Morand S, Colin ME, Belzunces LP (2002) Parasitism in the social bee
  \textit{Apis mellifera} : quantifying costs and benefits of behavioral
  resistance to \textit{Varroa destructor} mites. Apidologie 33(5):433--445,
  \doi{10.1051/apido:2002025},
  
\urlprefix\url{
http://www.apidologie.org/index.php?Itemid=129&option=com_article&access=doi&doi
=10.1051/apido:2002025&type=pdf},
  00000

\end{thebibliography}

\end{document}